\documentclass[a4paper,10pt]{article}
\pdfoutput=1
\usepackage{amsmath, amsthm, amssymb}
\usepackage{graphicx,color} 
\usepackage{algorithmicx}
\usepackage{algorithm}
\usepackage[noend]{algpseudocode}
\usepackage{authblk}
\usepackage{comment}

\theoremstyle{plain} 
\newtheorem{theorem}{Theorem}
\newtheorem{lemma}{Lemma}
\newtheorem{fact}{Fact}
\newtheorem{corollary}{Corollary}

\theoremstyle{remark}
\newtheorem{obs}{Observation}
\newtheorem{remark}{Remark}

\begin{document}

\title{Rooted Uniform Monotone Minimum Spanning Trees}

\author{Konstantinos Mastakas} \author{Antonios Symvonis}
\affil{School of Applied Mathematical and Physical Sciences\\ National Technical University of Athens, Greece \\ \{kmast,symvonis\}@math.ntua.gr}

\date{\today}

\maketitle

\begin{abstract}
We study the construction of the minimum cost spanning geometric graph of a given rooted point set $P$ where each point of $P$ is connected to the root by a path that satisfies a given property.
We focus on two properties, namely the monotonicity w.r.t.{} a single direction (\emph{$y$-monotonicity}) and the monotonicity w.r.t.{} a single pair of orthogonal directions (\emph{$xy$-monotonicity}).
We propose algorithms that compute the rooted $y$-monotone ($xy$-monotone) minimum spanning tree of $P$ in $O(|P|\log^2 |P|)$ (resp.{} $O(|P|\log^3 |P|)$) time when the direction (resp.{} pair of orthogonal directions) of monotonicity is given, and in $O(|P|^2\log|P|)$ time when the optimum direction (resp.{} pair of orthogonal directions) has to be determined.
We also give simple algorithms which, given a rooted connected geometric graph, decide if the root is connected to every other vertex by paths that are all monotone w.r.t.{} the same direction (pair of orthogonal directions). 
\end{abstract}

\section{Introduction}

A geometric path $W = (w_0$, $w_1$, \ldots, $w_l)$ is \emph{monotone in the direction of} $y$, also called $y-$\emph{monotone}, if it is $y$-\emph{decreasing}, i.e.{} $y(w_0) \geq y(w_1) \geq$ \ldots $\geq y(w_l)$ or if it is $y-$\emph{increasing}, i.e.{} $y(w_0) \leq y(w_1) \leq$ \ldots $\leq y(w_l)$, where $y(p)$ denotes the $y$ coordinate of a point $p$.
 $W$ is \emph{monotone} if there exists an axis $y'$ s.t.{} $W$ is $y'-$monotone.
Arkin et al.~\cite{ArkCM89} proposed a polynomial time algorithm which connects two given points by a geometric path that is monotone in a given (an arbitrary) direction and does not cross a set of obstacles, if such a path exists.
Furthermore, the problem of drawing a directed graph as an \emph{upward graph}, i.e.{} a directed geometric graph such that each directed path is $y-$increasing, has been studied in the field of graph drawing, e.g.{} see~\cite{BatT88,GarT01}.

A geometric graph $G = (P,E)$ is \emph{monotone} if every pair of points of $P$ is connected by a monotone geometric path, where the direction of monotonicity does not need to be the same for each pair.
If there exists a single direction of monotonicity, we denote the graph as \emph{uniform monotone}.
Uniform monotone graphs were also denoted as \emph{$1-$monotone graphs} by Angelini~\cite{Ang15}.
When the direction of monotonicity is known, say $y$, the graph is called \emph{$y-$monotone}.
Monotone graphs were introduced by Angelini et al.~\cite{AngCBFP12}.
The problem of drawing a graph as a monotone graph has been studied in the field of graph drawing; e.g.{} see~\cite{Ang15,AngCBFP12,AngDKMRSW15,HeH16}.
The reverse problem, namely, given a point set $P$ we are asked to construct a monotone spanning geometric graph on the points of $P$, has trivial solutions, i.e.{} the complete graph $K_{|P|}$ on the points of $P$ as well as the path graph $W_{|P|}$ which visits all points of $P$ in increasing order of their $y$ coordinates are both $y-$monotone spanning geometric graphs of $P$.

The \emph{Euclidean minimum spanning tree problem}, i.e.{} the problem of constructing the minimum cost spanning geometric tree of a plane point set $P$ (where the cost of the tree is taken to be the sum of the Euclidean lengths of its edges), has also received attention~\cite{preparata1988computational}. 
Shamos and Hoey~\cite{ShaH75} showed that it can be solved in $\Theta(|P|\log|P|)$ time. 

Combining the Euclidean minimum spanning tree problem with the notion of monotonicity leads to a large number of problems that, to the best of our knowledge, have not been previously investigated. 
The most general problem can be stated as follows: \emph{``Given a point set $P$ find the minimum cost monotone spanning geometric graph of $P$, i.e.{} the geometric graph such that every pair of points of $P$ is connected by a monotone path''.}
Since in a monotone graph the direction of monotonicity need not be the same for all pairs of vertices, it is not clear whether the minimum cost monotone spanning graph is a tree.
We call this problem the \emph{Monotone Minimum Spanning Graph problem.}
We note that there exist point sets for which the Euclidean minimum spanning tree is not monotone and hence does not coincide with the monotone minimum spanning graph.
Consider for example the point set with four points depicted in Figure~\ref{fig:EMSTvsMonotone} for which the Euclidean minimum spanning tree is a geometric path that is not monotone.

\begin{figure}[hbtp]
\centering
\includegraphics[scale=0.5]{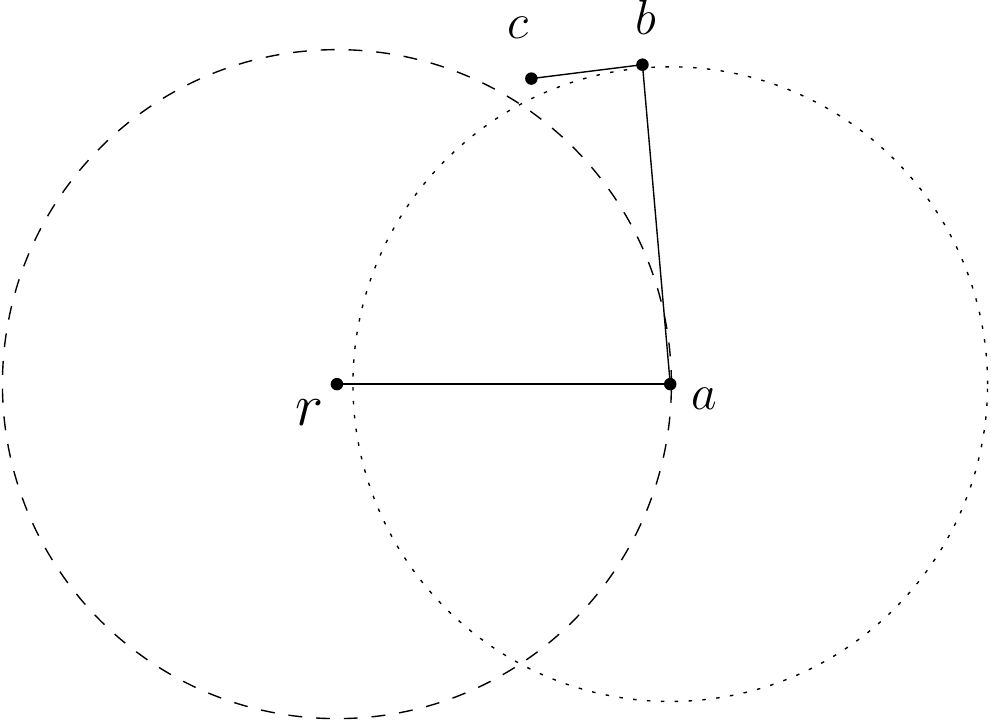}
\caption{The Euclidean minimum spanning tree which is depicted is not monotone}
\label{fig:EMSTvsMonotone}
\end{figure}

We focus on a simple variant of the general monotone minimum spanning graph problem. 
Let $P$ be a rooted  point set, i.e.{} a point set having a designated point, say $r$, as its root.
We do not insist on having monotone paths between every pair of points of $P$ but rather only between the root $r$ with all other points of $P$. 
Moreover, we insist that all paths are \emph{uniform} in the sense that they are all monotone with respect to the same direction, i.e.{} we build \emph{rooted uniform monotone graphs}.
Actually, as it turns out (Corollary~\ref{cor:y-tree}), in this problem the sought graphs are trees and, thus, we refer to it as the \emph{rooted Uniform Monotone Minimum Spanning Tree} (for short, \emph{rooted UMMST}) \emph{problem}.
In the rooted UMMST problem we have the freedom to select the direction of monotonicity.
When we are restricted to have monotone paths in a specific direction, say $y$, we have the \emph{rooted $y$-Monotone Minimum Spanning Tree} (for short, \emph{rooted $y$-MMST}) \emph{problem}.
Figure~\ref{fig:illYrSpanning}(a) illustrates a rooted $y-$monotone spanning graph of a rooted point set $P$, while the rooted $y-$MMST of $P$ is given in Figure~\ref{fig:illYrSpanning}(b).

\begin{figure}[htbp]
\begin{minipage}{0.45\textwidth}
\centering
\includegraphics[width=0.45\textwidth,keepaspectratio]{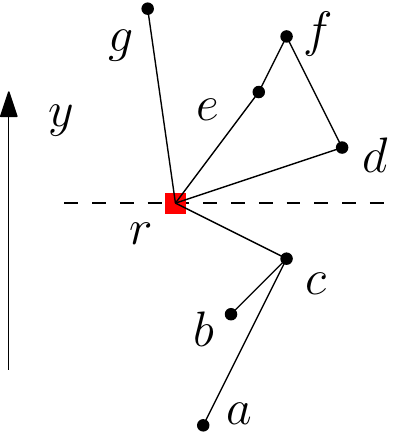}  
 \\ (a)
\end{minipage}
\hfill
\begin{minipage}{0.45\textwidth}
\centering
\includegraphics[width=0.45\textwidth,keepaspectratio]{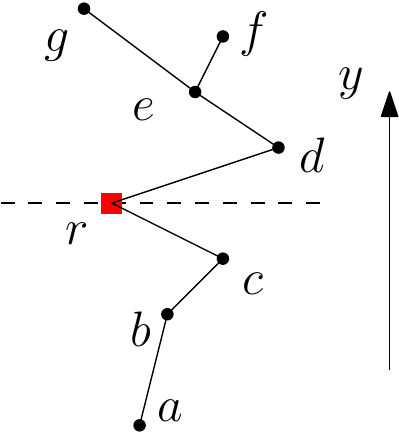}
\\ (b)
\end{minipage}

\caption{Illustration of rooted $y-$monotone spanning graphs
}
\label{fig:illYrSpanning}
\end{figure}

Rooted point sets have been previously studied in the context of minimum spanning trees.
The \emph{capacitated minimum spanning tree} is a tree that has a designated vertex $r$ (its root) and each of the subtrees attached to $r$ contains no more than  $c$ vertices.
$c$ is called the \emph{tree capacity}.
Solving the capacitated minimum spanning tree problem optimally has been shown by Jothi and Raghavachari to be NP-hard~\cite{Jothi:2005:AAC:1103963.1103967}.
In the same paper, they have also presented approximation algorithms for the case where the vertices correspond to points on the Euclidean plane. 

If the geometric path $W$ is both $x-$monotone and $y-$monotone then it is denoted as $xy-$\emph{monotone}.
Furthermore, if there exists a Cartesian System $x'y'$ s.t.{} $W$ is $x'y'-$monotone then $W$ is \emph{2D-monotone}.
Based on $xy-$monotone geometric paths and in analogy to the (rooted) monotone, uniform monotone and $y-$monotone graphs, we define the (\emph{rooted}) \emph{2D-monotone}, \emph{uniform 2D-monotone} and \emph{$xy-$monotone graphs}.
2D-monotone paths/graphs were also recently denoted by Bonichon et al.~\cite{BonBCKLV16} as \emph{angle-monotone paths/graphs}.
Bonichon et al.~\cite{BonBCKLV16} gave a $O(|P|\cdot|E|^2)$ time algorithm that decides if a geometric graph $G = (P,E)$ is 2D-monotone.
In order to do so, Bonichon et al.~\cite{BonBCKLV16} gave a $O(|E|^2)$ time algorithm which is used as a subroutine and decides if the graph is rooted 2D-monotone, where the root is a specified vertex. 
Bonichon et al.~\cite{BonBCKLV16} also noted that it is not always feasible to construct a planar 2D-monotone spanning geometric graph of a given point set.
Similarly to the rooted UMMST and rooted $y$-MMST problems we define the corresponding \emph{rooted Uniform 2D-Monotone Minimum Spanning Tree}  (for short, \emph{rooted 2D-UMMST}) and \emph{rooted $xy$-Monotone Minimum Spanning Tree}  (for short, \emph{rooted $xy$-MMST}) problems, which ask for the minimum cost rooted Uniform 2D-Monotone spanning tree and rooted $xy$-monotone spanning tree of a given rooted point set, respectively. 

A path/curve $W$ is \emph{increasing-chord} (see~\cite{LarM72,Rot94}) if for every four points $p_1, p_2, p_3, p_4$ traversed in this order along it, it holds that $d(p_2, p_3) \leq d(p_1, p_4)$ where $d(p,q)$ denotes the Euclidean distance between the points $p$ and $q$. 
A geometric graph $G = (P,E)$ is \emph{increasing-chord} if each two points of $P$ are connected by an increasing-chord path.
Increasing-chord graphs were introduced by Alamadari et al.~\cite{AlaCGLP13}.
Alamdari et al.~\cite{AlaCGLP13} noted that any 2D-monotone path/graph is also increasing-chord.
Drawing a graph as an increasing-chord graph is studied in~\cite{AlaCGLP13,NolPR16}. 
On the other hand, constructing increasing-chord graphs that span a given point set is studied in~\cite{AlaCGLP13,DehFG15,MasS15}.
In all the papers that construct increasing-chord spanning graphs of a point set, i.e.{} in~\cite{AlaCGLP13,DehFG15,MasS15}, the constructed increasing-chord paths connecting the vertices are additionally 2D-monotone.

\textbf{Our contribution:}

Let $P$ be a rooted point set. 
We give algorithms that produce the rooted $y-$MMST of $P$ and the rooted $xy-$MMST of $P$ in $O(|P|\log^2 |P|)$ time and $O(|P|\log^3 |P|)$ time, respectively. 
We also propose algorithms that build the rooted UMMST of $P$ and the rooted 2D-UMMST of $P$ in $O(|P|^2 \log |P|)$ time when the optimum direction and the optimum pair of directions has to be determined, respectively.
For all these four problems, we provide a $\Omega(|P| \log |P|)$ time lower bound which is easily derived.

We also propose simple algorithms that decide whether a given connected geometric graph on a rooted point set is (i) rooted $y-$monotone, (ii) rooted uniform monotone, (iii) rooted $xy-$monotone and (iv) rooted uniform 2D-monotone.

\section{Definitions and Preliminaries}

In this article we deal with the Euclidean plane, i.e.{} every point set is a subset of $\mathbb{R}^2$, and we consider only rooted point sets.
Let $x,y$ be the axes of a Cartesian System.
The $x$ and $y$ coordinates of a point $p$ are denoted by $x(p)$ and $y(p)$, respectively.
W.l.o.g.{}, we assume that the root $r$ of the point sets coincides with the origin of the Cartesian System, i.e.{}, $x(r) = y(r) = 0$.
We also assume that the point sets are \emph{in general position}, i.e.{} no three points are collinear.

Let $P$ be a point set.
$P$ is called \emph{positive} (\emph{negative}) w.r.t.{} the direction of $y$ or $y-$positive ($y-$negative) if for each $p \in P$, $y(p) \geq 0$ (resp.~$y(p) \leq 0$). 
 Let $a$ be a real number.
 By $P_{y \leq a}$ we denote the set of points of $P$ that have $y$ coordinate less than or equal to $a$.
Subsets $P_{y \geq a}$, $P_{x \leq a}$, $P_{x \geq a}$, $P_{|y| \leq a}$, $P_{|y| \geq a}$, $P_{|x| \leq a}$ and $P_{|x| \geq a}$ are similarly defined.
$P_{x \leq a,y \leq a}$ denotes the set $P_{x \leq a}\cap P_{y \leq a}$.
Subsets $P_{x \geq a,y \geq a}$, $P_{x \leq a, y \geq a}$ and $P_{x \geq a, y \leq a}$ are defined similarly.
The Euclidean plane is divided into four quadrants, i.e.{} the quadrants $\mathbb{R}^2_{x \geq 0, y \geq 0}$, $\mathbb{R}^2_{x \leq 0, y \geq 0}$, $\mathbb{R}^2_{x \leq 0, y \leq 0}$ and $\mathbb{R}^2_{x \geq 0, y \leq 0}$.

Let $P$ be a point set and $p$ be a point of the plane, then $d(p,P)$ denotes the Euclidean distance from $p$ to the point set $P$, i.e.{} $d(p,P) = \min\limits_{q \in P} d(p,q)$.

The line segment with endpoints $p$ and $q$ is denoted as $\overline{pq}$.
The \emph{slope} of a line $L$ is the angle that we need to rotate the $x$ axis counterclockwise s.t.{} the $x$ axis becomes parallel to $L$. 
Each slope belongs to the range $[0,\pi)$.

A \emph{geometric graph} $G = (P,E)$ consists of a set $P$ of points which are denoted as its vertices and a set $E$ of line segments with endpoints in $P$ which are denoted as its edges.
If $P$ is rooted then $G$ is a \emph{rooted geometric graph}.
The cost of a geometric graph $G = (P,E)$, denoted as $\text{cost}(G)$, is the sum of the Euclidean lengths of its edges. 
Let $G_1 = (P_1, E_1), G_2 = (P_2, E_2), \ldots, G_n = (P_n ,E_n)$ be $n$ geometric graphs then the union $G_1\cup G_2\cup \ldots \cup G_n$ is the geometric graph $G = (P, E)$ s.t.{} $P = P_1 \cup P_2 \cup \ldots \cup P_n$ and $E = E_1 \cup E_2 \cup \ldots \cup E_n$.

The \emph{closest point} (or \emph{nearest neighbor}) \emph{problem} is an important problem in computational geometry.
It was initially termed as the \emph{post-office problem} by Knuth~\cite{Knu73}. 
In this problem there exists a set $S$ of points that is static (it cannot be changed by inserting points to it or deleting points from it) and the goal is to find the closest point (or nearest neighbor) from $S$ to a given query point.
This problem is usually reduced to the problem of locating in which region of a planar subdivision the query point is located~\cite{Sha75,DobL76}.
Efficient static data structures have been constructed to answer these queries in logarithmic time by performing fast preprocessing algorithms, e.g.{} see~\cite{LipT80,Kir83}.
Concerning the semi-dynamic version of the closest point problem, in which insertions of points to $S$ are allowed,
Bentley~\cite{Ben79} gave a very useful semi-dynamic data structure. 

\begin{fact}[Bentley~\cite{Ben79}]\label{fact:BentleyStructure}
There exists a semi-dynamic data structure that allows only two operations, the insertion of a point and a closest point query.
Where, a closest point query takes $O(\log ^2 n)$ time (with $n$ denoting the size of the structure) and inserting $n$ elements in the structure takes $O(n \log ^2 n)$ total time.
\end{fact}

A variant of the closest point problem that was studied in the fourth section of the article of Bentley~\cite{Ben79}, comes in handy when we study the rooted $xy-$MMST problem. 
More specifically, in this variant, each point in $S$ is associated with a numerical attribute value.
Given a query point $q$, the goal is to find the closest point to $q$ from the subset of points of $S$ for which their attribute value belongs to a specified range.
For the static version, i.e.{} when $S$ cannot be mutated, Bentley~\cite{Ben79} gave a static data structure that adds a multiplicative logarithmic factor in the query and preprocessing time of the static data structure for the closest point problem.
Concerning the semi-dynamic version of this problem, Bentley~\cite{Ben79} implicitly produced, i.e{} from his corresponding static data structure (section 4 of~\cite{Ben79}) and his results about the decomposable problems (section 3 of~\cite{Ben79}), a very useful semi-dynamic data structure.

\begin{fact}[Bentley~\cite{Ben79}]\label{fact:rangeClosest}
There exists a semi-dynamic data structure that supports the operations (i) insertion of a point, and (ii) the closest point to a query point from the subset of points in the data structure whose attribute value is within a given range.
Where, a closest point query takes $O(\log ^3 n)$ time (with $n$ denoting the number of points in the data structure) and inserting $n$ points in the data structure takes $O(n \log ^3 n)$ total time.
\end{fact}

\section{The Rooted $y$-Monotone Minimum Spanning Tree (rooted $y-$MMST) Problem}\label{sec:y-uMMST}

In this section we study the construction of the rooted $y-$MMST of a rooted point set $P$.
We initially show that we can deal with $P_{y \leq 0}$ and $P_{y\geq 0}$ separately.
Then, we provide a characterization of the rooted $y-$MMST of rooted $y-$positive (or $y-$negative) point sets.
Using the previous two, we develop an algorithm that constructs the rooted $y-$MMST of $P$.
We also provide a lower bound for the time complexity of any algorithm that solves the rooted $y-$MMST problem as well as a simple recognition algorithm for rooted $y-$monotone graphs. 

We recall that we assume that the root $r$ of a rooted point set is located at the point $(0,0)$.

\begin{obs}\label{obs:unnesEdge}
Let $P$ be a rooted point set and $G = (P,E)$ be a rooted $y-$monotone spanning graph of $P$ and let $\overline{p_dp_u} \in E$ with $y(p_d) < 0 < y(p_u)$. 
Then, every path from the root $r$ to a point $p \in P\setminus\{r\}$ that contains $\overline{p_dp_u}$ is not $y-$monotone since it moves ``south'' to $p_d$ and then ``north'' to $p_u$, or vice versa.
\end{obs}

\begin{corollary}\label{cor:unEdges}
Let $P$ be a rooted point set, $G^{\text{opt}} = (P,E)$ be the  rooted $y-$monotone minimum spanning graph of $P$ and
$p_d,p_u \in P$ such that $y(p_d) < 0 < y(p_u)$.
Then, $\overline{p_dp_u} \notin E$.
\end{corollary}

Corollary~\ref{cor:unEdges} implies that the root $r$ of a rooted point set $P$ splits the problem of finding the rooted $y-$monotone minimum spanning graph of $P$ into two independent problems.
Namely, producing the rooted $y-$monotone minimum spanning graph of (i) $P_{y\leq 0}$, and (ii) $P_{y\geq 0}$.
Hence, we obtain the following Lemma.

\begin{lemma}\label{lem:splitUpDown}
Let $P$ be a rooted point set and $G^{\text{opt}}$ be the rooted $y-$monotone minimum spanning graph of $P$.
Furthermore, let $G_{y \leq 0}^{\text{opt}}$ and $G_{y \geq 0}^{\text{opt}}$ be the rooted $y-$monotone minimum spanning graphs of $P_{y \leq 0}$ and $P_{y \geq 0}$, respectively.
Then, $G^{\text{opt}}$ is the union of $G_{y \leq 0}^{\text{opt}}$ and $G_{y \geq 0}^{\text{opt}}$.
\end{lemma}

We now study the construction of the rooted $y-$monotone minimum spanning graph of a rooted $y-$positive (or $y-$negative) point set $P$ with root $r$.
The main idea is to connect each $p \in P \setminus\{r\}$ to its closest point with absolute $y$ coordinate less than the absolute $y$ coordinate of $p$.
However, since we have assumed that there might be (at most) two points with the same $y$ coordinate the analysis becomes a little more complicated.
We define $S[P,y]$ to be the sequence  of points of $P$ ordered by the following rule: \emph{``The points of $S[P,y]$ are ordered w.r.t.{} their absolute $y$ coordinates and, if two points have the same $y$ coordinate, then they are ordered w.r.t.{} their distance from the preceding points in $S[P,y]$.''}.
More formally, $S[P,y] = (r = p_0,p_1,p_2, \ldots, p_n)$ s.t.{} $|y(p_0)| \leq |y(p_1)| \leq |y(p_2)| \leq \ldots \leq |y(p_n)|$ and $|y(p_i)| = |y(p_{i+1})|$ implies that $d(p_i, \{p_0, p_1, \ldots, p_{i-1}\}) \leq d(p_{i+1}, \{p_0, p_1, \ldots, p_{i-1}\})$ and $P = \{p_0,p_1,p_2, \ldots, p_n\}$.
We now give a characterization of the rooted $y-$monotone minimum spanning graph of $P$. 

\begin{lemma}\label{lem:yMonotoneChar}
Let $G=(P,E)$ be a rooted geometric graph where $P$ is a rooted $y-$positive (or $y-$negative) point set with $S[P,y]$ $= (r$ $= p_0,p_1,p_2$ , \ldots, $p_n)$.  
Then, $G$ is the rooted $y-$monotone minimum spanning graph of $P$ if and only if (i) $p_n$ is connected in $G$ only with its closest point (or nearest neighbor) from $\{p_0$, $p_1$, \ldots, $p_{n-1}\}$, i.e.{} the point $p_j$ such that $d(p_n,p_j) = d(p_n$, $\{p_0$, $p_1$, \ldots, $p_{n-1}\})$, and (ii) $G\setminus\{p_n\}$ is the rooted $y-$monotone minimum spanning graph of $P\setminus\{p_n\}$.
\end{lemma}

\begin{proof}
The ($\Rightarrow$) direction can be easily proved by contradiction.
We now prove the ($\Leftarrow$) direction. 
 $G$ is a rooted $y-$monotone graph that spans $P$ since $G\setminus\{p_n\}$ is a rooted $y-$monotone graph that spans $P\setminus\{p_n\}$ and $p_n$ is connected with another point in the graph. 
Let $G^{\text{opt}}$ be the rooted $y-$monotone minimum spanning graph of $P$.
Then, $p_n$ is connected in $G^{\text{opt}}$ with some other point in $P$, and thus $G^{\text{opt}}$ contains an edge of cost at least $d(p_n, \{p_0, p_1, \ldots, p_{n-1}\})$.
Furthermore, the graph $G^{\text{opt}} \setminus\{p_n\}$ is also rooted $y-$monotone, hence its cost is at least the  cost of the rooted $y-$monotone minimum spanning graph of $P\setminus\{p_n\}$.
Thus, $G^{\text{opt}}$ has at least the same cost as $G$.  
\end{proof}

Lemma~\ref{lem:splitUpDown} and Lemma~\ref{lem:yMonotoneChar} lead to the next Corollary.

\begin{corollary}\label{cor:y-tree} 
The rooted $y-$monotone minimum spanning graph of a rooted point set $P$ is a geometric tree.
\end{corollary}

Let $P$ be a rooted $y-$positive (or $y-$negative) point set and $S[P,y] =$ $(r =$ $p_0$, $p_1$, \ldots, $p_n)$.
We call the closest point to $p_i$ from $\{p_0, p_1, \ldots, p_{i-1}\}$ the \emph{parent} of $p_i$ and we denote it as \emph{par$(p_i)$}.
More formally, par$(p_i) = p_j$ if and only if $p_j \in \{p_0, p_1, \ldots, p_{i-1}\}$ and $d(p_i, p_j) = d(p_i, \{p_0, p_1, \ldots, p_{i-1}\})$.
Then, Lemma~\ref{lem:yMonotoneChar} implies the following Corollary.

\begin{corollary}\label{cor:parentCharact}
The edges of the rooted $y-$MMST of $P$ are exactly the line segments $\overline{\text{par}(p_i)p_i},$ for  $i =1, 2, \ldots, n$.
\end{corollary}

Corollary~\ref{cor:parentCharact} implies a $O(|P|^2)$ time algorithm for producing the rooted $y-$MMST of $P$.
However, using the semi-dynamic data structure for closest point queries given by Bentley~\cite{Ben79}, the time complexity of our rooted $y-$MMST algorithm becomes $O(|P|\log^2|P|)$.
Our rooted $y-$MMST algorithm is described in Algorithm~\ref{alg:GIyMonotone}.

\begin{algorithm}[!hbt]
\caption{rooted $y-$MMST}\label{alg:GIyMonotone}
\textbf{Input:} A rooted $y-$positive (or $y-$negative) point set $P$ with root $r$. 

\textbf{Output:} The rooted $y$-Monotone Minimum Spanning Tree of $P$.
\begin{algorithmic}[1]
\State $T \gets$ the geometric graph with $P$ as its vertex set and $\emptyset$ as its edge set.  
\State Sort the points of $P$ w.r.t.{} their absolute $y$ coordinates, constructing the sequence $S$.
If some points tie, use an arbitrary order (their relative order may be changed later on, during the execution of the algorithm, since $S$ is used in order to calculate the $S[P,y]$, i.e.{} at the end of the algorithm $S$ equals to $S[P,y]$). More formally, $S$ = $(r$ = $S[0]$, $S[1]$, \ldots, $S[n])$ with $|y(S[i])| \leq |y(S[i+1])|$. 
\State ProximityDS is an (initially empty) semi-dynamic data structure for closest point queries. 
\State Insert $r$ in ProximityDS.
\State $i \gets 1$ \Comment{\parbox[t]{0.75\textwidth}{\emph{$i$ indicates the index of the next point we need to process, i.e.{} find its parent and then insert it in the ProximityDS}}} 
\While{$i \leq n$}
  \State par$(S[i]) \gets $ closest point to $S[i]$ from ProximityDS, obtained by performing a closest point query in ProximityDS.
  \newline 
  \Comment{\emph{Check if $S[i]$ has the correct order in $S$, i.e.{} its order in $S$ is the same as its order in $S[P,y]$. This might not be true if $y(S[i])$ is equal to $y(S[i+1])$}}
  \If{$y(S[i]) \neq y(S[i+1])$} \Comment{\emph{$S[i]$ has the correct order in $S$}} 
  \State insert $S[i]$ in ProximityDS and insert the edge $\overline{\text{par}(S[i])S[i]}$ in $T$.
  \State $i \gets i + 1$
  \Else \Comment{\parbox[t]{0.75\textwidth}{\emph{$y(S[i]) = y(S[i+1])$, hence we do not know if the points with indices $i$ and $i+1$ in $S$ have the correct order in $S$}}}
  \State par$(S[i+1]) \gets $ closest point to $S[i+1]$ from ProximityDS, obtained by performing a closest point query in ProximityDS.\newline
  \Comment{\emph{Check if $S[i]$ and $S[i+1]$ have the correct order in $S$ or have to be swapped.}}
  \If{$d(\text{par}(S[i]), S[i]) > d(\text{par}(S[i+1]), S[i+1])$} 
  \State Swap the points with indices $i$ and $i+1$ in $S$ 
    \EndIf
  \newline  \Comment{\emph{Check if par($S[i+1]$) is $S[i]$}}     
      \If{$d(S[i], S[i+1]) < d(\text{par}(S[i+1]), S[i+1]) $}         \State par$(S[i+1]) \gets S[i]$. 
      \EndIf 
      \State insert $S[i]$ and $S[i+1]$ in ProximityDS
    \State Insert the edges $\overline{\text{par}(S[i])S[i]}$ and $\overline{\text{par}(S[i+1])S[i+1]}$ in $T$. \newline 
    \Comment{\emph{Both $S[i]$ and $S[i+1]$ were processed}} 
  \State $i \gets i + 2$  
  \EndIf
\EndWhile
\State \textbf{return } $T$.
\end{algorithmic}
\end{algorithm}

\begin{theorem}\label{thm:yRootedMonotoneMain}
The rooted $y-$MMST  of a rooted point set $P$ can be computed in $O(|P|\log^2 |P|)$ time.
\end{theorem}

\begin{proof}
We first construct $P_{y\leq 0}$ and $P_{y\geq 0}$.
Then, we apply our rooted $y-$MMST algorithm  
(Algorithm~\ref{alg:GIyMonotone})
 on $P_{y\leq 0}$ and $P_{y\geq 0}$ constructing $T_{y \leq 0}$ and $T_{y \geq 0}$, respectively.
By Corollary~\ref{cor:parentCharact}, $T_{y \leq 0}$ and $T_{y \geq 0}$ are the rooted $y-$MMSTs of $P_{y\leq 0}$ and $P_{y\geq 0}$, respectively.
Using Fact~\ref{fact:BentleyStructure} and since $O(|P|)$ insertions and $O(|P|)$ closest point queries are performed, computing $T_{y \leq 0}$ and $T_{y \geq 0}$ takes $O(|P|\log^2 |P|)$ time.
By Lemma~\ref{lem:splitUpDown}, $T_{y \leq 0} \cup T_{y \geq 0}$ is the rooted $y-$MMST of $P$. 
\end{proof}

In the next Theorem, we give a lower bound for the time complexity of any algorithm which given a rooted point set $P$ produces the rooted $y-$MMST of $P$. 

\begin{theorem}\label{thm:lowerBoundyMonotone}
Any algorithm which given a rooted point set $P$, produces the rooted $y-$MMST of $P$ requires $\Omega (|P|\log |P|)$ time.
\end{theorem}

\begin{proof}
We use the reduction from sorting that was given by Shamos~\cite{Sha78}.
Let $(a_1$, $a_2$, \dots, $a_n)$ be a sequence of nonnegative integers. 
We reduce this sequence to the rooted point set $P = \{r= (0,0)$, $(a_1, a_1^2)$, $(a_2, a_2^2)$, \ldots, $(a_n, a_n^2)\}$.
Then, the rooted $y-$MMST of $P$ contains exactly the edges $\overline{rp_1}$, $\overline{p_1p_2}$, \ldots ,$\overline{p_{n-1}p_n}$ s.t.{} $a_i' = x(p_i), i = 1,2, \ldots, n$, where $(a_1', a_2', \ldots, a_n')$ is the sorted permutation of $(a_1, a_2, \dots, a_n)$.
The lower bound follows since sorting $n$ numbers requires $\Omega (n\log n)$ time.
\end{proof}

We note that using the same reduction, i.e.{} the reduction from sorting that was given by Shamos~\cite{Sha78}, the same lower bound can be easily obtained for the rooted UMMST and (rooted) Monotone Minimum Spanning Graph Problem.

We conclude this section, by showing that rooted $y-$monotone graphs can be efficiently recognized.
Our approach is similar to the approach employed in the third section of the article of Arkin et al.~\cite{ArkCM89}. 

\begin{theorem}\label{thm:recDirected}
 Let $G = (P,E)$ be a rooted connected geometric graph. 
 Then, we can decide in $O(|E|)$ time if $G$ is rooted $y-$monotone.
\end{theorem}

\begin{proof}
We first transform $G$ into a directed geometric graph $\overrightarrow{G}$ in $O(|E|)$ time, by assigning direction to the edges and removing some of them.
Let $\overline{pq}$ be an edge of $G$.
If $p$ and $q$ belong to opposite half planes w.r.t.{} the $x$ axis then $\overline{pq}$ cannot be used in a $y-$monotone path from the root $r$ to a point of $P\setminus \{r\}$ (see Observation~\ref{obs:unnesEdge}).
Hence, we remove the edge $\overline{pq}$ from the graph.
If $y(p) = y(q)$ then we insert both $\overrightarrow{pq}$ and $\overrightarrow{qp}$ in $\overrightarrow{G}$.
Otherwise, assuming w.l.o.g., that  $|y(p)| < |y(q)|$, we insert $\overrightarrow{pq}$ in $\overrightarrow{G}$.
$G$ is rooted $y-$monotone if and only if $r$ is connected with all other points of $P$ in $\overrightarrow{G}$.
The latter can be easily decided in $O(|E|)$ time by a breadth first search or a depth first search traversal. 
\end{proof}

\section{Rooted Uniform Monotone Graphs: Minimum Spanning Tree Construction, and Recognition}\label{sec:generalRootMon}

In this section we study the rooted uniform monotone graphs.
We initially study the computation of the rooted UMMST of a rooted point set (Subsection~\ref{subsec:UMMST}) and then we deal with the recognition of rooted uniform monotone graphs (Subsection~\ref{subsec:recUM}).

\subsection{Building the Rooted UMMST}\label{subsec:UMMST}

In this subsection, we focus on the rooted UMMST problem.
In contrast with the rooted $y-$MMST problem where the direction $y$ of monotonicity was given, here we are asked to determine the optimum direction of monotonicity, say $y'$, and the corresponding rooted $y'-$MMST.
We tackle the problem by giving a rotational sweep algorithm. 
Rotational sweep is a well known technique in computational geometry in which a (directed) line is rotated counterclockwise (or clockwise) and during this rotation important information about the solution of the problem is updated.
We note that crucial to the development of our algorithm for the rooted UMMST problem is the fact that it is sufficient to take into account only a quadratic number of directions of monotonicity  (to be proved in Lemma~\ref{lem:MonotoneNecessaryDirections}).
 This fact is based on the following observations.

\begin{obs}\label{obs:changeInRotation}
Let $y'$ be an axis.
If we rotate the $y'$ axis counterclockwise then the sequence $S[P_{y' \geq 0}, y']$ or the sequence $S[P_{y' \leq 0}, y']$ changes only when the $y'$ axis reaches (moves away from) a line perpendicular to a line passing through two points of $P$.
Then, by Lemma~\ref{lem:yMonotoneChar} and Corollary~\ref{cor:parentCharact}, the rooted $y'-$MMST of $P$ may only change at the same time.
\end{obs}

\begin{obs}\label{obs:slopeLessThanPi}
Let $y'$ and $y''$ be axes of opposite directions. 
Then, the rooted $y'-$MMST  of $P$ is the same as the rooted  $y''-$MMST of $P$. 
Hence, when computing the rooted UMMST of $P$ we only need to take into account the $y'$ axes such that the angle that we need to rotate the $x$ axis counterclockwise to become codirected with the $y'$ is less than $\pi$.
\end{obs}

Based on Observations~\ref{obs:changeInRotation} and~\ref{obs:slopeLessThanPi}, 
 we define the set $\Theta = \{\theta \in [0,\pi): \theta$ is the slope of a line perpendicular to a line passing through two points of $P\}$. 
We also define $S[\Theta]$ to be the sorted sequence that contains the slopes of $\Theta$ in increasing order, i.e.{} $S[\Theta] = (\theta_0, \theta_1, \ldots, \theta_{m-1})$, $\theta_i < \theta_{i+1}$, $i =0, 1, \ldots, m-2$ and $m \leq \binom{|P|}{2}$.
We further define the set $\Theta_{\text{critical}} = \{\theta_0$, $\theta_1$, \ldots, $\theta_{m-1}\} \cup \{$ $\frac{\theta_0 + \theta_1}{2}$, $\frac{\theta_1 + \theta_2}{2}$, \ldots, $\frac{\theta_{m-2} + \theta_{m-1}}{2}$, $\frac{\theta_{m-1} + \pi}{2}\}$
which we call the \emph{critical set of slopes} since, as we show in Lemma~\ref{lem:MonotoneNecessaryDirections},
examining the axes with slope in $\Theta_{\text{critical}}$ is  sufficient for computing the rooted UMMST of $P$.
$|\Theta_{\text{critical}}| = O(|P|^2)$. 
We now assign ``names'' to  the axes with slopes in $\Theta_{\text{critical}}$. 
Let $y_{2i}$ be the axis with slope $\theta_i$, $i = 0, 1, \ldots, m-1$ and let $y_{2i + 1}$ be the axis with slope $\frac{\theta_i + \theta_{i+1}}{2}$, $i = 0, 1, \ldots, m-2$, and $y_{2m-1}$ be the axis with slope $\frac{\theta_{m-1} + \pi}{2}$.
Note that the subscript of each axis gives its order when the axes are sorted w.r.t.{} their slope.
Note also that axes with even subscripts correspond to lines perpendicular to lines passing through two points of $P$.

\begin{lemma}\label{lem:MonotoneNecessaryDirections}
The rooted UMMST of $P$ is one of the rooted $y'-$MMST  of $P$ over all axes $y'$ with slope in $\Theta_{\text{critical}}$ and, more specifically, the one of minimum cost.  
\end{lemma}

\begin{proof}
Let $y'$ and $y''$ be axes of slope $\theta'$ and $\theta''$, respectively, such that $\theta_i < \theta', \theta'' < \theta_{i+1}$ for some $0 \leq i \leq m - 2$ and let $T'$, $T''$ be the rooted  $y'-$MMST  and the  rooted $y''-$MMST of $P$, respectively.
By Observation~\ref{obs:changeInRotation}, cost$(T') = $cost$(T'')$.
As a result, we need to take into account only one of the $y'$ axes of slope $\theta '$ with $\theta_i < \theta' < \theta_{i+1}$.
We take into account the axis $y_{2i+1}$.
Additionally, we take into account the axis $y_{2m-1}$ for which the rooted $y_{2m-1}-$MMST of $P$ is the same as the rooted $y'-$MMST  of $P$ for any $y'$ axis with slope in the range $(\theta_{m-1},\pi) \cup [0,\theta_0)$.

We also need to take into account all the axes of slope $\theta_{i}$, i.e.{} the axes $y_{2i}$, $ 0 \leq i \leq m-1$, since, by Observation~\ref{obs:changeInRotation}, the rooted $y_{2i}-$MMST of $P$ might differ from all the other rooted $y'-$MMST  of $P$, $y' \neq y_{2i}$ (the sorted sequence $S[P_{y_{2i} \geq 0},y_{2i}]$ (or $S[P_{y_{2i} \leq 0},y_{2i}]$) might be different from every other $S[P_{y' \geq 0},y']$ (resp.{}, $S[P_{y' \leq 0},y']$) ).
See, for example, Figure~\ref{fig:consAxes}. 
\end{proof}

\begin{figure}[htbp]
\centering
\begin{minipage}{0.28\textwidth}
\centering
\includegraphics[height = 0.12\textheight, keepaspectratio]{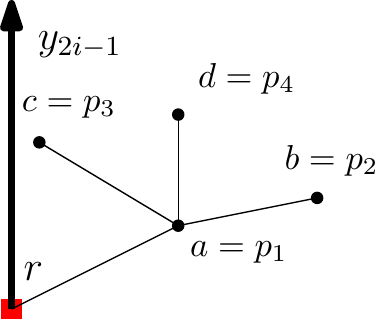} \\
(a)
\end{minipage}
\hfill
\begin{minipage}{0.28\textwidth}
\centering
\includegraphics[height = 0.12\textheight, keepaspectratio]{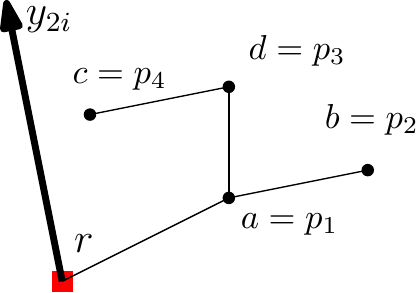} \\ (b)
\end{minipage}
\hfill
\begin{minipage}{0.28\textwidth}
\centering
\includegraphics[height = 0.12\textheight, keepaspectratio]{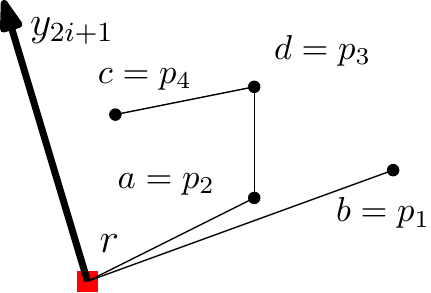} \\(c)
\end{minipage}
\caption{Illustration of three consecutive axes $y_{2i-1}, y_{2i}$ and $y_{2i+1}$ with slope in $\Theta_{\text{critical}}$ for the point set $\{r,a,b,c,d\}$ where $y_{2i}$ is perpendicular to the lines through the pairs of points $(a, b)$ and $(c,d)$.
In (a) $S[P_{y_{2i-1} \geq 0}, y_{2i-1}] = (r,a = p_1$, $b = p_2$, $c = p_3$, $d = p_4)$ while in (b) $S[P_{y_{2i} \geq 0}, y_{2i}] = (r, a, b, d, c)$ and in (c) $S[P_{y_{2i+1} \geq 0}, y_{2i+1}] = (r, b, a, d, c)$.
}
\label{fig:consAxes}
\end{figure}

We now describe our algorithm which produces the rooted UMMST of a rooted point set $P$.
Our rooted UMMST algorithm is a rotational sweep algorithm.
It considers an axis $y'$, which initially coincides with $y_0$, and then it rotates it counterclockwise until $y'$  becomes opposite to the $x$ axis.
Throughout this procedure, it updates the rooted $y'-$MMST of $P$.
By Lemma~\ref{lem:MonotoneNecessaryDirections}, it only needs to obtain each rooted $y_i-$MMST of $P$, where $y_i$ is an axis with slope in $\Theta_{\text{critical}}$, $0 \leq i \leq 2m-1$.
Let $T^{\text{opt}}_i$ be the rooted $y_i-$MMST of $P$, $0 \leq i \leq 2m-1$.
Our rooted UMMST algorithm can now be stated as follows:
It initially constructs $T^{\text{opt}}_0$ using Theorem~\ref{thm:yRootedMonotoneMain}. 
Then, it iterates for $i = 1, 2, \ldots, 2m - 1$ obtaining at the end of each iteration $T^{\text{opt}}_{i}$ by modifying $T^{\text{opt}}_{i-1}$.
In order to do this efficiently, it maintains a tree $T$ which is initially equal to $T^{\text{opt}}_0$ and throughout its operation it evolves to $T^{\text{opt}}_{1}$, $T^{\text{opt}}_{2}$, \ldots, $T^{\text{opt}}_{2m-1}$.
 Similarly, it maintains the sequences $S^-$ and $S^+$ which are initially equal to $S[P_{y_{0} \leq 0}, y_{0}]$ and $S[P_{y_{0} \geq 0}, y_{0}]$, respectively, and evolve to $S[P_{y_{i} \leq 0}, y_{i}]$ and $S[P_{y_{i} \geq 0}, y_{i}]$, respectively, $i =1, 2, \ldots, 2m-1$.
Our algorithm stores the axis which corresponds to the produced rooted UMMST  of $P$, so far, in the variable ``minAxis''.
In its final step, it recomputes the rooted {``minAxis''}-MMST of $P$ using Theorem~\ref{thm:yRootedMonotoneMain} and returns this tree.
The pseudocode of our rooted UMMST algorithm is presented in Algorithm~\ref{alg:rootedMonotoneAlgo}.

\begin{algorithm}
\caption{rooted UMMST}\label{alg:rootedMonotoneAlgo}
\textbf{Input:} A rooted point set $P$.

\textbf{Output:} The rooted  Uniform Monotone Minimum Spanning Tree of $P$.
\begin{algorithmic}[1]
\State Compute the axes $y_0, y_1, \ldots, y_{2m-1}$ with slopes in $\Theta_{\text{critical}}$.
\State Compute $T^{\text{opt}}_0$, $S[P_{y_0 \leq 0}, y_0]$ and $S[P_{y_0 \geq 0}, y_0]$ using Theorem~\ref{thm:yRootedMonotoneMain}.
\State $T \gets T^{\text{opt}}_0$, $S^- \gets S[P_{y_0 \leq 0}, y_0]$, $S^+ \gets S[P_{y_0 \geq 0}, y_0]$, minCost$\gets$ cost$(T^{\text{opt}}_0)$ and minAxis$\gets y_0$. 
\For{ $i \gets 1$ to $2m-1$}
    \State Update $T, S^-, S^+$ such that $T$ equals to $T^{\text{opt}}_{i}$ and $S^-$ (resp.~$S^+$) equals to $S[P_{y_{i} \leq 0}, y_{i}]$ (resp.~$S[P_{y_{i} \geq 0}, y_{i}]$).
  \If{cost$(T) < $minCost}
  \State minCost $\gets$ cost$(T)$ and minAxis$\gets y_i$.
\EndIf
\EndFor
\State \textbf{return } the minAxis-MMST  of $P$, computed  using Theorem~\ref{thm:yRootedMonotoneMain}.
\end{algorithmic}
\end{algorithm}

\begin{theorem}\label{thm:spanningTreeOneDir}
The rooted UMMST of a rooted point set $P$ can be computed in $O(|P|^2 \log |P|)$ time.
\end{theorem}

\begin{proof}
By Lemma~\ref{lem:MonotoneNecessaryDirections}, our rooted UMMST algorithm  (Algorithm~\ref{alg:rootedMonotoneAlgo})
produces the rooted UMMST of $P$.
We now show that its time complexity is $O(|P|^2 \log |P|)$.
The axes $y_0$, $y_1$, \ldots, $y_{2m-1}$ with slopes in $\Theta_{\text{critical}}$ can be computed in $O(|P|^2 \log |P|)$ time.
Let $k_i$ be the number of pairs of points of $P$ that have the same projection onto the $y_{2i}$ axis, $0 \leq i \leq m-1$.
Then $\sum\limits_{i =0}^{m-1} k_i = \binom{|P|}{2}$.
For each $i = 0$, $1$, \ldots, $m-1$, we compute a list $L_i$ which contains these $k_i$ pairs.
All $L_i, 0\leq i \leq m-1$, can be computed in  $O(|P|^2 \log |P|)$ total  time.

For each point $p$ in $S^-\setminus \{r\}$ (resp.~in $S^+\setminus \{r\}$) we maintain a data structure PD$(p)$ which is a self-balancing binary search tree that contains all the points that precede $p$ in $S^-$ (resp.~$S^+$) accompanied with their distance from $p$.
More formally, let $S^-$ be equal to $(r = p_0, p_1, \ldots, p_s)$.
Then, for each $p_j, j = 1, 2, \ldots, s$, PD$(p_j)$ contains the pairs  $(p_0, d(p_0, p_j))$, $(p_1, d(p_1, p_j))$, \ldots, $(p_{j-1}, d(p_{j-1}, p_j))$.
The key of each $(p_l, d(p_l, p_j)), l = 0, 1, \ldots, j-1$, is the distance $d(p_l, p_j)$.
Similarly, we define PD$(p)$ for each $p \in S^+\setminus \{r\}$.
We employ these PD$(p)$, $p \in P\setminus \{r\}$, data structures since using the information stored in them we can obtain the parent of each $p$ efficiently. 
In more detail, for each $p \in P\setminus \{r\}$ the par$(p)$ in $T$ can be obtained or updated in $O(\log |P|)$ time by taking into account the PD$(p)$, since the pair $(\text{par}(p), d(\text{par}(p),p))$ is the element with the minimum key in PD$(p)$.

Computing the initial values of $T, S^-,S^+$ and PD$(p), p \in P\setminus \{r\}$ can be done in $O(|P|^2 \log |P|)$ time.
This is true since $T^{\text{opt}}_0$, $S[P_{y_0 \leq 0}, y_0]$ and $S[P_{y_0 \geq 0}, y_0]$ are computed in $O(|P|\cdot \log ^2 |P|)$ time (see Theorem~\ref{thm:yRootedMonotoneMain}).
Furthermore, computing PD$(p_j)$ for some $p_j$ in $S^-\setminus \{r\}$ (resp.~in $S^+\setminus \{r\}$), when $S^-$ (resp.~$S^+$) equals to $S[P_{y_0 \leq 0}, y_0]$ (resp.~$S[P_{y_0 \geq 0}, y_0]$) takes $O(|P| \log |P|)$ time since we have to insert each $(p_i,d(p_i,p_j)$, with $i < j$, to PD$(p_j)$ and each such insertion takes $O(\log |P|)$ time.
Hence, the total running time for initially computing all PD$(p),p \in P\setminus \{r\}$, is $O(|P|^2 \log |P|)$.

Let $T$ be equal to $T^{\text{opt}}_{i-1}$ and let $S^-$ (resp.~$S^+$) be equal to $S[P_{y_{i-1} \leq 0},y_{i-1}]$ (resp.~$S[P_{y_{i-1} \geq 0},$ ~$ y_{i-1}]$) then $T$ and $S^-$ (resp.~$S^+$) can be updated such that $T$ becomes equal to $T^{\text{opt}}_{i}$ and $S^-$ (resp.~$S^+$) becomes equal to  $S[P_{y_{i} \leq 0},y_{i}]$ (resp.~$S[P_{y_{i} \geq 0},y_{i}]$) in: 
\begin{enumerate}
\item\label{en:first} $O(k_{\lfloor\frac{i}{2}\rfloor}\log |P|)$ time if $i$ is even and $y_i$ is not perpendicular to a line passing through the root $r$ and another point in $P$.
\item \label{en:second} $O(k_{\lfloor\frac{i}{2}\rfloor}\log |P|)$ time if $i$ is odd and $y_{i-1}$ is not perpendicular to a line passing through the root $r$ and another point in $P$.
\item \label{en:third} $O(|P| \log |P|)$ time if $i$ is even and $y_i$ is perpendicular to a line passing through the root $r$ and another point $q \in P\setminus \{r\}$.
\item\label{en:fourth} $O(|P| \log |P|)$ time if $i$ is odd and $y_{i-1}$ is perpendicular to a line passing through $r$ and another point $q \in P\setminus \{r\}$.
\end{enumerate}

We first explain how to maintain the data structures so that all points relevant to case~\ref{en:first} are treated in $O(k_{\lfloor\frac{i}{2}\rfloor}\log |P|)$ time.
Since $y_i$ is not perpendicular to a line connecting the root $r$ with another point in $P$ then $P_{y_{i} \leq 0} = P_{y_{i-1} \leq 0}$ and $P_{y_{i} \geq 0} = P_{y_{i-1} \geq 0}$.
Hence, no point is inserted into (or removed from) $S^-$ or $S^+$.
However, some points which had different projections onto the $y_{i-1}$ axis, now have the same projection onto the $y_{i}$ axis.

We only explain how to deal with the points in $S^+$; the points in $S^-$ can be treated similarly. 
 Recall that $L_{\lfloor\frac{i}{2}\rfloor}$ contains the $k_{\lfloor\frac{i}{2}\rfloor}$ pairs of points of $P$ that have the same projection onto the $y_i$ axis.
Let $S^+ = (r = p_0, p_1, \ldots, p_s)$ and let $(p_{j_1}, p_{j_1+1})$, $(p_{j_2}, p_{j_2+1})$, \ldots, $(p_{j_k}, p_{j_k+1})$ with $j_1 < j_2 < \ldots < j_k$ be the $k$ pairs of points in $S^+$ that belong to $L_{\lfloor\frac{i}{2}\rfloor}$, i.e.{} they are connected by line perpendicular to $y_i$. 
Then, each $p \in S^+\setminus \{p_{j_1}$, $p_{j_1+1}$, $p_{j_2}$, $p_{j_2+1}$, \ldots, $p_{j_k}$, $p_{j_k+1}\}$ has the same relative order with the other points of $S^+$ w.r.t.{} both the $y_{i}$ axis and the $y_{i-1}$ axis. 
Hence, $p$ is placed at the correct position in $S^+$, the parent of $p$ in $T$ is correct (see Lemma~\ref{lem:yMonotoneChar}) and PD$(p)$ does not need to be updated.

As a result, the only changes that may be necessary regard the points in $\{p_{j_1}$, $p_{j_1+1}$, $p_{j_2}$, $p_{j_2+1}$, \ldots, $p_{j_k}$, $p_{j_k+1}\}$.
For these points we may need to recalculate their parent in $T$.
We may also need to swap some consecutive $p_{j_l}$ and  $p_{j_l+1}, l =1,2, \ldots, k$, in $S^+$.
Finally, we may need to update some of the PD$(p_{j_l})$ and PD$(p_{j_l+1}), l =1,2, \ldots, k$.

For each $l = 1$ to $k$ we do the following: 

We compute $d(p_{j_l+1}, \{p_0, p_1, \ldots, p_{j_l-1}\})$ in $O(\log |P|)$ time using PD$(p_{j_l+1})$.

If $d(p_{j_l+1},$ $\{p_0, p_1$, \ldots, $p_{j_l-1}\}) \geq$ $d(p_{j_l},$ par$(p_{j_l}))$, then we do not update anything, since the points $p_{j_l}$ and $p_{j_l+1}$ are placed at the correct position in $S^+$.

If, on the other hand, $d(p_{j_l+1},$ $\{p_0, p_1$, \ldots, $p_{j_l-1}\}) <$ $d(p_{j_l},$ par$(p_{j_l}))$, we remove the edges $\overline{\text{par}(p_{j_l})p_{j_l}}$ and $\overline{\text{par}(p_{j_l+1})p_{j_l+1}}$ from $T$. 
Then, we swap the order of the points $p_{j_l+1}$ and $p_{j_l}$ in $S^+$, i.e.{} if $S^+$ was previously equal to $(r = p_0,$ $p_1,$ \ldots, $p_{j_l}, p_{j_l+1},$ \ldots, $p_n)$, now $S^+$ becomes equal to $(r = p_0$, $p_1$, \ldots, $p'_{j_l}, p'_{j_l+1}$, \ldots, $p_{n})$ with $p'_{j_l}$ equal to $p_{j_l+1}$ and $p'_{j_l+1}$ equal to $p_{j_l}$.
We then insert the pair $(p'_{j_l}, d(p'_{j_l}, p'_{j_l+1}))$ into PD$(p'_{j_l+1})$ and remove the pair $(p'_{j_l+1}, d(p'_{j_l+1}, p'_{j_l}))$ from PD$(p'_{j_l})$.
Finally, in $T$ we connect the point $p'_{j_l}$ (resp.~$p'_{j_l+1}$) with the point $p$ s.t.{} $(p,d(p,p'_{j_l}))$ (resp.~$(p,d(p,p'_{j_l+1}))$) has the minimum key in PD$(p'_{j_l})$ (resp.~PD$(p'_{j_l+1})$) and update its parent accordingly.
Using PD$(p'_{j_l})$ and PD$(p'_{j_l+1})$ all this process which concerns a single pair of points is completed in $O(\log |P|)$ time.
Thus, $k_{\lfloor\frac{i}{2}\rfloor}$ pairs of points are treated in $O(k_{\lfloor\frac{i}{2}\rfloor} \log |P|)$ time.

Case~\ref{en:second} is treated similarly to Case~\ref{en:first}.

We now treat Case~\ref{en:third}.
First, observe that this case occurs exactly $|P|-1$ times, i.e.{} one time for each pair $(r,p), p \in P\setminus\{r\}$. 
In this case, either $P_{y_{i-1} \leq 0}$ is a strict subset of (not equal to) $P_{y_{i} \leq 0}$ or $P_{y_{i-1} \geq 0}$ is a strict subset of (not equal to) $P_{y_{i} \geq 0}$.
W.l.o.g.{} we assume that $P_{y_{i-1} \geq 0}$ is a strict subset of $P_{y_{i} \geq 0}$ and that $y_{i-1}(q) < 0$, i.e.{} $q$ belongs to $S^-$ w.r.t.{} $y_{i-1}$ while it belongs to both $S^+$ and $S^-$ w.r.t.{} the $y_i$ axis.

We first explain how to deal with the insertion of $q$ into the point set $P_{y_{i} \geq 0}$. 
We insert $q$ into the sequence $S^+$ right after $r$.
We do not need to update the PD$(q)$ since it already contains only the pair $(r,d(r,q))$.
We also do not need to update par$(q)$.
Then, for each point $p \in S^+$ with $p \notin \{r,q\}$ we insert the pair $(q, d(q,p))$ in the PD$(p)$ and if $d(q,p)$ is the lowest key in PD$(p)$, then we remove $\overline{\text{par}(p)p}$ from $T$, we assign $q$ to the par$(p)$ and then we insert the edge $\overline{pq}$ to $T$.
We are now done with $q$.
All the previously described actions take $O(|P|\log |P|)$ time, for the single pair $(r,q)$.  

Then, for the pairs of points that have the same projection onto the $y_{i}$ axis, except for $(r,q)$, we apply the procedure described in Case~\ref{en:first}.  
As shown in that Case, this is done in $O(k_{\lfloor\frac{i}{2}\rfloor} \log |P|) = O(|P| \log |P|)$ time, since $ k_{\lfloor\frac{i}{2}\rfloor} < |P|$. 

Case~\ref{en:fourth} is treated similar to Case~\ref{en:third}.

Since $\sum\limits_{i =0}^{m-1} k_i = O(|P|^2)$, the total running time of our algorithm is $O(|P|^2 \log |P|)$. 
\end{proof}

\subsection{Recognizing Rooted Uniform Monotone Graphs}\label{subsec:recUM}

We now proceed to the problem of deciding if a given rooted connected geometric graph is rooted uniform monotone.
Like we did for the rooted UMMST problem, we tackle this decision problem with a rotational sweep algorithm.
Again, it is sufficient to take into account only polynomially many directions and more specifically linearly many (to be proved in Lemma~\ref{lem:recArbNecDirections}).

We first define some auxiliary sets.
Let $G = (P, E)$ be a rooted connected geometric graph with root $r$ and $p$ be a point of $P\setminus \{r\}$. 
Let $A(p,y)$ be the set that contains all the adjacent points  to $p$ that are on the same side with $p$ w.r.t.{} the $x$ axis and are strictly closer to the $x$ axis than $p$. 
More formally, $A(p,y) = \{q: q\in$ Adj$(p)$, $q$ lies on the same half plane with $p$ w.r.t.{} the $x$ axis and $|y(q)| < |y(p)|$ $\}$.
Let $B(y)$ denote the set $\{p: p\in P\setminus\{r\}$ and $A(p,y) \neq \o \}$. 
Let $C(y)$ be the set that consists of the points $p\in P\setminus \{r\}$ that (i) do not belong to $B(y)$ and (ii) are connected with some other point $q$ with the same $y$ coordinate such that $A(q,y) \neq \o$.
More formally, $C(y) = \{p: p\in P\setminus (B(y)\cup\{r\})$ such that there exists $q \in $Adj$(p)$ with $y(q) = y(p) $ and $A(q,y) \neq \o \}$. 
An example of a rooted geometric graph and the corresponding sets is given in Figure~\ref{fig:illSets}.

\begin{figure}[htbp]
\centering
\includegraphics[height = 0.14\textheight,keepaspectratio]{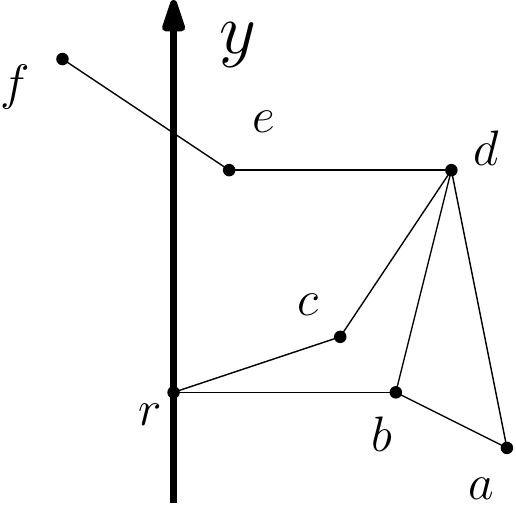}
\caption{
$A(a,y) = \{b\}, A(b,y) = \o, A(c,y) = \{r\}, A(d,y) = \{b,c\}, A(e,y) = \o$ and $A(f,y) = \{e\}$.
$B(y) = \{a,c,d,f\}$ and $C(y) = \{e\}$.}
\label{fig:illSets}
\end{figure}

We now give a characterization of rooted $y-$monotone graphs based on the previously defined sets. 

\begin{lemma}\label{lem:charYrootSetTerm}
Let $G = (P,E)$ be a rooted connected geometric graph such that for each $p \in P\setminus\{r\}$, $y(p) \neq 0$. 
Then, $G$ is rooted $y-$monotone if and only if $|B(y)| +$ $|C(y)| = |P| - 1$.
\end{lemma}

\begin{proof}
We first prove the ($\Rightarrow$) direction.
We prove that each $p \in P\setminus\{r\}$ is included in exactly one of $B(y)$ and $C(y)$ (i.e.{} $P\setminus\{r\} = B(y) \cup C(y)$) which is equivalent to $|B(y)| +$ $|C(y)|$ = $|P| - 1$.
Each $p \in P\setminus\{r\}$ is connected with $r$ by a $y-$monotone path, hence there exists a point $q \in P\setminus \{p\}$ on the same side with $p$ w.r.t.{} the $x$ axis that is adjacent to $p$ and $|y(q)| \leq |y(p)|$.
If for such a point $q$ it holds that $|y(q)| < |y(p)|$, then $A(p,y) \neq \o$ and hence $p \in B(y)$.
Otherwise, there exists exactly one point $q \in P\setminus \{r,p\}$ that is adjacent to $p$ and on the same side with $p$ w.r.t.{} the $x$ axis and $|y(q)| = |y(p)|$.
Then, for $q$ it holds that $y(q) = y(p)$ and hence $p$ belongs to $C(y)$. 
The ($\Leftarrow$) direction can be easily proved by induction on the number of points. 
\end{proof}

\begin{remark}\label{rem:generCharYrootSetTerm}
If there exists a point $p \in P\setminus\{r\}$ with $y(p) = 0$ then  $G$ is rooted $y-$monotone if and only if (i) $p$ is connected with $r$ and (ii) $|B(y)| +$ $|C(y)|$ equals to $|P| - 2$.
\end{remark}

\begin{remark}\label{rem:extToZeroCoord}
If we know $B(y)$, $C(y)$ and whether there exists a point $p \in P\setminus \{r\}$ with $y(p) = 0$ connected to $r$, we can decide if $G$ is rooted $y-$monotone.
This implies a $O(|E|)$ time algorithm, different from the one given in Theorem~\ref{thm:recDirected}, which recognizes rooted $y-$monotone geometric graphs.
\end{remark}

\begin{obs}\label{obs:changeInSets}
Let $G = (P,E)$ be a rooted connected geometric graph. 
If we rotate an axis $y'$ counterclockwise, then $B(y')$, $C(y')$ and the points $p$ in $P\setminus \{r\}$ with $y'(p) = 0$, change only when the $y'$ axis reaches (or moves away from) a line that is perpendicular to an edge of $G$ or that is perpendicular to a line passing through the root $r$ and a point of $P\setminus\{r\}$.
\end{obs}

Using similar arguments to the ones employed for solving the rooted UMMST problem, we define a set of critical slopes and appropriate axes which we have to test in order to decide if the rooted connected geometric graph $G = (P,E)$ is rooted uniform monotone.
Let $\Theta = \{\theta \in [0,\pi): \theta$ is the slope of a line perpendicular to either an edge of $G$ or to a line passing through the root $r$ and another point of $P\}$.
$S[\Theta]$ is the sorted sequence that contains the slopes of $\Theta$ in increasing order, i.e.{} $S[\Theta] = (\theta_0, \theta_1, \ldots, \theta_{m-1})$, $\theta_i < \theta_{i+1}$, $i =0, 1, \ldots, m-2$ and $m < |E| +|P|$.
We define the critical set of slopes, $\Theta_{\text{critical}} = \{\theta_0$, $\theta_1$, \ldots, $\theta_{m-1}\} \cup \{$ $\frac{\theta_0 + \theta_1}{2}$, $\frac{\theta_1 + \theta_2}{2}$, \ldots, $\frac{\theta_{m-2} + \theta_{m-1}}{2}$, $\frac{\theta_{m-1} + \pi}{2}\}$. 
We now assign ``names'' to the axes with slope in $\Theta_{\text{critical}}$.
Let $y_{2i}$ be the axis with slope $\theta_i, 0\leq i \leq m-1$.
Moreover, let $y_{2i+1}$ be the axis of slope $\frac{\theta_i + \theta_{i+1}}{2}$, $0\leq i \leq m-2$ and $y_{2m-1}$ be the axis of slope $\frac{\theta_{2m-1} + \pi}{2}$.
In analogy to Lemma~\ref{lem:MonotoneNecessaryDirections} we obtain the next Lemma.

\begin{lemma}\label{lem:recArbNecDirections}
$G$ is rooted uniform monotone if and only if it is rooted $y'-$monotone for some $y'$ axis of slope in $\Theta_{\text{critical}}$.
\end{lemma}

We now give a rotational sweep algorithm that tests whether a given rooted connected geometric graph $G = (P,E)$ is rooted uniform monotone.
Our rooted uniform monotone recognition algorithm rotates an axis $y'$ which initially coincides with $y_0$ until it becomes opposite to the $x$ axis.
Throughout this rotation, it checks if $G$ is rooted $y'-$monotone.
Taking into account Lemma~\ref{lem:recArbNecDirections}, our algorithm only needs to test if $G$ is rooted $y_i-$monotone for some $i = 0,1, \ldots, 2m-1$.
The pseudocode of our rooted uniform monotone recognition algorithm is given in Algorithm~\ref{alg:recRooted}.

\begin{algorithm}[htbp]
\caption{rooted uniform monotone recognition} \label{alg:recRooted}
\textbf{Input:} A rooted connected geometric graph $G = (P, E)$.

\textbf{Output:} The axis of monotonicity if it exists, otherwise, null.
\begin{algorithmic}[1]
\State axis $\gets $ null.
\State compute the axes $y_0, y_1, \ldots, y_{2m-1}$ with slopes in $\Theta_{\text{critical}}$.
\If{$G$ is rooted  $y_0-$monotone} axis $\gets y_0$. 
    \EndIf
\For{ $i \gets 1$ to $2m -1$}
    \If{$G$ is rooted $y_i-$monotone} axis $\gets y_i$. 
    \EndIf
\EndFor
\State \textbf{return} axis
\end{algorithmic}
\end{algorithm}

\begin{theorem}\label{thm:recRootedMonotone}
Let $G = (P,E)$ be a rooted connected geometric graph.
Then, we can decide in $O(|E|\log |P|)$ time if $G$ is rooted uniform monotone.
\end{theorem}

\begin{proof}

By Lemma~\ref{lem:recArbNecDirections}, it is immediate that our rooted uniform monotone recognition algorithm 
 (Algorithm~\ref{alg:recRooted})
 decides if $G$ is rooted uniform monotone.
We now show that its time complexity is $O(|E|\log |P|)$.
Computing the axes $y_0, y_1$, \ldots, $y_{2m-1}$, with slope in $\Theta_{\text{critical}}$ can be done in $O(|E| \log |P|)$ time. 
Let $k_i$ be the number of pairs of points of $P$ connected by an edge perpendicular to $y_{2i}, 0 \leq i \leq m-1$.
Then, $\sum\limits_{i = 0}^{m-1} k_i = |E|$.
For each $i = 0,1$,\ldots,$m-1$, we construct a list $L_i$ containing the $k_i$ pairs of points of $P$ that are connected by an edge perpendicular to $y_{2i}$.
All $L_i$, $0 \leq i \leq m-1$, can be computed in $O(|E|\log |P|)$ total time.

Let $y_i$ be the last axis taken into account. 
Our algorithm maintains for each $p \in P\setminus\{r\}$ a data structure $A(p)$ which represents the set $A(p,y_i)$ (which is a subset of the Adj$(p)$).
$A(p)$ contains the indices of the points of $P$ that belong to $A(p,y_i)$.
$A(p)$ can be implemented by any data structure which supports insert, delete and retrieve operations in $O(\log |P|)$ time (e.g. a $2-3$ tree).
Our algorithm also maintains the data structure $B$ that represents the $B(y_i)$. 
In order to performing the insert and delete operations in $O(1)$ time, $B$ is implemented as an array of boolean with size $O(|P|)$. 

Let $p$ be a point in $P\setminus \{r\}$ then computing the $A(p)$ s.t.{} $A(p)$ equals to $A(p, y_0)$ takes $O(|\text{Adj}(p)|\log |P|)$ time since all the points $q$ adjacent to $p$ are checked and if $q$ belongs to the same half plane with $p$ w.r.t.{} the $x$ axis and $|y(q)| < |y(p)|$ then $q$ is inserted into $A(p)$ in $O(\log |P|)$ time.
Hence, computing all $A(p), p \in P\setminus \{r\}$, takes $O(|E|\log |P|)$ total time.
Given each $A(p), p \in P\setminus \{r\}$, computing $B$ s.t.{} $B$ equals to $B(y_0)$ takes $O(|P|)$ time, i.e.{} $O(1)$ time to check if $A(p) \neq \o$ for each $p \in P\setminus \{r\}$. 

When we move from the $y_{i-1}$ to the $y_i$ axis, the necessary updates we need to make s.t.{} for each $p \in P\setminus\{r\}$, $A(p)$ represents the set $A(p,y_i)$ and $B$ represents the $B(y_i)$ take:
\begin{enumerate}
\item $O(k_{\lfloor \frac{i}{2}\rfloor} \log |P|)$ time if $i$ is even and there is no point $p \in P\setminus \{r\}$ with the same projection with $r$ onto $y_i$
\item $O(k_{\lfloor \frac{i}{2}\rfloor} \log |P|)$ time if $i$ is odd and there is no point $p \in P\setminus \{r\}$ with the same projection with $r$ onto $y_{i-1}$
\item $O((|\text{Adj}(q)| + k_{\lfloor \frac{i}{2}\rfloor}) \log |P|)$ time if $i$ is even and $y_i$ is perpendicular to the line passing through $r$ and the point $q \in P\setminus \{r\}$
\item $O( (|\text{Adj}(q)| + k_{\lfloor \frac{i}{2}\rfloor}) \log |P|)$ time if $i$ is odd and $y_{i-1}$ is perpendicular to the line passing through $r$ and the point $q \in P\setminus \{r\}$

\end{enumerate}

Cases (1),(2),(3) and (4) are proved similarly to the Cases (1),(2),(3) and (4) in Theorem~\ref{thm:spanningTreeOneDir}.

We also note that given $B(y_{2i})$ and each $A(p,y_{2i}), p\in P\setminus \{r\}$, then computing $C(y_{2i}), 0 \leq i \leq m-1$, takes $O(k_i)$ time using the list $L_i$.
Furthermore, if we have both $B$ equal to $B(y_i)$ and $C$ equal to $C(y_i)$ and know if $y_i$ is perpendicular to some line passing through $r$ and another point $p$ of $P$ with $\overline{pr} \in E$, then we can test if $G$ is $y_i-$rooted-monotone in $O(1)$ time (see Lemma~\ref{lem:charYrootSetTerm} and Remark~\ref{rem:extToZeroCoord}). 

Since, $\sum\limits_{i = 0}^{m-1} k_i = |E|$ and $\sum\limits_{q \in P\setminus\{r\}} |\textrm{Adj}(q)| = O(|E|)$, the time complexity of the algorithm is $O(|E| \log |P|)$.
\end{proof}

We note that the approach we took for deciding if a given rooted connected geometric graph is rooted uniform monotone has some similarities with the approach employed in the third section of the article of Arkin et al.~\cite{ArkCM89}.

\section{Rooted Uniform 2D-monotone Graphs: Minimum Spanning Tree Production, and Recognition}\label{sec:2D-UMMST}

In this section we study monotonicity w.r.t.{} two perpendicular axes. 
Our treatment is analogous to that of Section~\ref{sec:y-uMMST} and Section~\ref{sec:generalRootMon}.

\subsection{The Rooted $xy-$MMST Problem}
We first study the case where the perpendicular axes are given, i.e.{} they are the $x$ and $y$ axes.

In analogy with Observation~\ref{obs:unnesEdge} we obtain the following Observation.
\begin{obs}\label{obs:xyUnnesEdges}
Let $P$ be a rooted point set, $G = (P,E)$ be a rooted $xy-$monotone spanning graph of $P$ and let $\overline{pp'} \in E$ where either (i) $p$ and $p'$ lie on different quadrants of the plane or (ii) $p$ and $p'$ lie on the same quadrant of the plane and $(|x(p) | - |x(p')|)(|y(p)| - |y(p')|) < 0$. 
Then, every path from the root $r$ to a point $q \in P\setminus\{r\}$ that contains $\overline{pp'}$ is not $xy-$monotone.
\end{obs}

\begin{corollary}\label{cor:pointsOnDifQuadr} 
Let $P$ be a rooted point set and $G^{\text{opt}}$ be the rooted $xy-$monotone minimum spanning graph of $P$. 
Let $p$ and $q$ be points of $P$ that do not lie on the same quadrant of the plane, then $G^{\text{opt}}$ does not contain the $\overline{pq}$.
\end{corollary}

The previous Corollary implies that producing the rooted $xy-$monotone minimum spanning graph of $P$ can be split into four independent problems.
More specifically, obtaining the rooted $xy-$monotone minimum spanning graph of (i) $P_{x \leq 0,y \leq 0}$, (ii) $P_{x \geq 0,y \geq 0}$, (iii) $P_{x \leq 0, y \geq 0}$ and (iv) $P_{x \geq 0, y \leq 0}$.
This is stated more formally in the following Lemma.

\begin{lemma}\label{lem:splitQuadrants}
Let $P$ be a rooted point set and $G^{\text{opt}}$ be the rooted $xy-$monotone minimum spanning graph of $P$. 
Let $G_{x \leq 0,y \leq 0}^{\text{opt}}$, $G_{x \geq 0,y \geq 0}^{\text{opt}}$, $G_{x \leq 0, y \geq 0}^{\text{opt}}$ and $G_{x \geq 0, y \leq 0}^{\text{opt}}$ be the rooted $xy-$monotone minimum spanning graph of $P_{x \leq 0,y \leq 0}$, $P_{x \geq 0,y \geq 0}$, $P_{x \leq 0, y \geq 0}$ and $P_{x \geq 0, y \leq 0}$, respectively.
Then, $G^{\text{opt}}$ is $G_{x \leq 0,y \leq 0}^{\text{opt}} \cup G_{x \geq 0,y \geq 0}^{\text{opt}}\cup G_{x \leq 0, y \geq 0}^{\text{opt}} \cup G_{x \geq 0, y \leq 0}^{\text{opt}}$.
\end{lemma}

Let $P$ be a rooted point set confined to one quadrant of the plane.
Then, we define $S[P,y,x]$ to be the sequence that consists of the points of $P$, such that the points in $S[P,y,x]$ are ordered w.r.t.{} their absolute $y$ coordinates and if two points have the same absolute $y$ coordinate, then they are ordered w.r.t.{} their absolute $x$ coordinates.
More formally, $S[P,y,x] = (r= p_0,p_1,p_2, \ldots, p_n)$ such that $|y(p_0)| \leq |y(p_1)| \leq |y(p_2)| \leq \ldots \leq |y(p_n)|$ and $|y(p_i)| = |y(p_{i+1})|$ implies that $ |x(p_i)| < |x(p_{i+1})|$ and $P = \{p_0,p_1,p_2, \ldots, p_n\}$.
Using similar arguments with the proof of Lemma~\ref{lem:yMonotoneChar} we obtain a characterization of the
rooted $xy-$monotone minimum spanning graph of $P$.

\begin{lemma}\label{lem:minCostxy2DChar}
Let $P$ be a rooted point set confined to one quadrant of the plane, $S[P,y,x]$ $=$ $(r=$ $p_0$, $p_1$, $p_2$, \ldots, $p_n)$ and $G$ be a geometric graph with vertex set $P$.
Then, $G$ is the rooted $xy-$monotone minimum spanning graph of $P$ if and only if (i) $p_n$ is connected only with its closest point with absolute $x$ coordinate less than or equal to $|x(p_n)|$ from $\{p_0$, $p_1$, \ldots, $p_{n-1}\}$, i.e.{} the point $p_j \in P\setminus \{p_n\}$ such that $|x(p_j)| \leq |x(p_n)|$ and $d(p_n,p_j) = d(p_n, P_{|x| \leq |x(p_n)|}\setminus \{p_n\})$, and (ii) $G\setminus\{p_n\}$ is the rooted $xy-$monotone minimum spanning graph of $P\setminus\{p_n\}$.
\end{lemma}

Lemma~\ref{lem:splitQuadrants} and Lemma~\ref{lem:minCostxy2DChar} lead to the next Corollary.

\begin{corollary}
The rooted $xy-$monotone minimum spanning graph of a rooted point set $P$ is a geometric tree.
\end{corollary}

Let $P$ be a rooted point set confined to one quadrant of the plane and $S[P,y,x]$ $=$ $(r=$ $p_0$, $p_1$, $p_2$, \ldots, $p_n)$.
For each $i = 1, 2, \ldots, n$, we call the closest point with absolute $x$ coordinate in the range $[0, |x(p_i)|]$ to $p_i$ from $\{p_0$, $p_1$, \ldots, $p_{i-1}\}$, the parent of $p_i$ and denote it as par$(p_i)$.
Equivalently, par$(p_i) = p_j$ if and only if $j < i$, $|x(p_j)| \leq |x(p_i)|$ and $d(p_i, p_j)$ $=$ $d(p_i$, $\{p_0$, $p_1$,\ldots, $p_{i-1}$$\}_{|x| \leq |x(p_i)|})$.
Then, Lemma~\ref{lem:minCostxy2DChar} yields the following Corollary.

\begin{corollary}\label{cor:xyEdgesChar}
The edges of the rooted $xy-$MMST of $P$ are exactly the line segments $\overline{\text{par}(p_i)p_i}$, $i = 1$, $2$, \ldots, $n$. 
\end{corollary}

The previous Corollary implies a $O(|P|^2)$ time algorithm that produces the rooted $xy-$MMST of $P$. 
However, using the semi-dynamic data structure for closest point with attribute value in specified range queries that was implicitly produced by Bentley~\cite{Ben79}, the time complexity of our rooted $xy-$MMST algorithm is reduced to $O(|P| \cdot \log ^3 |P|)$.

\begin{theorem}\label{thm:minSpanningXYrooted}
The rooted $xy-$MMST of a rooted point set $P$ can be computed in $O(|P| \cdot \log ^3 |P|)$ time.
\end{theorem}

\begin{proof}
We initially construct $P_{x \leq 0,y \leq 0}$, $P_{x \geq 0,y \geq 0}$, $P_{x \leq 0, y \geq 0}$ and $P_{x \geq 0, y \leq 0}$.
Then, we apply our rooted $xy-$MMST algorithm on $P_{x \leq 0, y \leq 0}$, $P_{x \geq 0, y \geq 0}$, $P_{x \leq 0, y \geq 0}$ and $P_{x \geq 0, y \leq 0}$ and we obtain the trees $T_{x \leq 0, y \leq 0}^{\text{opt}}$, $T_{x \geq 0,y \geq 0}^{\text{opt}}$, $T_{x \leq 0, y \geq 0}^{\text{opt}}$ and $T_{x \geq 0, y \leq 0}^{\text{opt}}$, respectively. 
Using Fact~\ref{fact:rangeClosest} and since $O(|P|)$ insertions and $O(|P|)$ closest point with attribute value in specified range queries are performed, computing $T_{x \leq 0, y \leq 0}^{\text{opt}}$, $T_{x \geq 0,y \geq 0}^{\text{opt}}$, $T_{x \leq 0, y \geq 0}^{\text{opt}}$ and $T_{x \geq 0, y \leq 0}^{\text{opt}}$ takes $O(|P| \cdot \log ^3 |P|)$ time. 
Finally, we return the union of $T_{x \leq 0, y \leq 0}^{\text{opt}}$, $T_{x \geq 0,y \geq 0}^{\text{opt}}$, $T_{x \leq 0, y \geq 0}^{\text{opt}}$ and $T_{x \geq 0, y \leq 0}^{\text{opt}}$ which by Lemma~\ref{lem:splitQuadrants} is the rooted $xy-$MMST of $P$. 
\end{proof}

Using the same reduction that we used in Theorem~\ref{thm:lowerBoundyMonotone}, i.e.{} the reduction from sorting given by Shamos~\cite{Sha78}, we obtain a lower bound for the time complexity of every algorithm which solves the rooted $xy-$MMST (or the rooted 2D-UMMST or the rooted 2D-MMST) problem. 

\begin{theorem}
Producing the rooted $xy-$MMST (or the rooted 2D-UMMST or the rooted 2D-MMST) of a rooted point set $P$ requires $\Omega (|P|\log |P|)$ time.
\end{theorem}

We also give a linear time algorithm that recognizes rooted $xy-$monotone graphs.

\begin{theorem}\label{thm:recRootedXY}
Let $G = (P, E)$ be a rooted connected geometric graph.
Then, we can decide in $O(|E|)$ time if $G$ is rooted $xy-$monotone.
\end{theorem}

\begin{proof}
Our proof is similar to the proof of Theorem~\ref{thm:recDirected}.
We transform $G$ to a directed graph $\overrightarrow{G}$ in $O(|E|)$ time as follows.
Let $\overline{pq}$ be an edge of $G$.
If $p$ and $q$ lie on the same quadrant of the plane and $|x(p)| \leq |x(q)|$ and $|y(p)| \leq |y(q)|$ (resp.~$|x(q)| \leq |x(p)|$ and $|y(q)| \leq |y(p)|$) we direct $\overline{pq}$ from $p$ to $q$ (resp.~from $q$ to $p$).
By observation~\ref{obs:xyUnnesEdges}, it follows that all the other edges cannot be traversed by a $xy-$monotone path connecting $r$ with another point of $P$. 
Hence, we remove them.
Then, $G$ is rooted $xy-$monotone if and only if $r$ is connected with all other points of $P$ in $\overrightarrow{G}$.
We can decide the latter in $O(|E|)$ time by applying a breadth first search or a depth first search traversal. 
\end{proof}

\subsection{The Rooted 2D-UMMST Problem}

We now study the problem of computing the rooted 2D-UMMST of a given rooted point set $P$.
We give a $O(|P|^2 \log |P|)$ time rotational sweep algorithm that solves the problem, analogous to our rooted UMMST algorithm given in Subsection~\ref{subsec:UMMST}.

\begin{obs}\label{obs:rotationEventsXYproduction}
Let $x'y'$ be a Cartesian System. 
If we rotate the Cartesian System $x'y'$ counterclockwise, then the rooted $x'y'-$MMST of $P$ changes only when the $y'$ axis reaches (or moves away from) a line that is perpendicular or parallel to a line passing through two points of $P$.
\end{obs}

\begin{proof}
When the $y'$ axis reaches a line that is perpendicular or parallel to a line passing through two points of $P\setminus \{r\}$, say $p$ and $q$, then it might become feasible (while previously this was not feasible) to connect $p$ and $q$, with the line segment $\overline{pq}$, such that $p$ or $q$ traverses $\overline{pq}$ in a $x'y'-$monotone path from it to $r$.
In this case, by Corollary~\ref{cor:xyEdgesChar}, the rooted $x'y'-$MMST of $P$ may change.  
Furthermore, when the $y'$ becomes perpendicular or parallel to a line passing through $r$ and a point $p \in P\setminus \{r\}$ the point $p$ belongs to two quadrants (while previously belonged only to one of them). 
As a result, for some points $q_1$, $q_2$, \ldots, $q_k \in P\setminus\{r,p\}$ belonging to the new quadrant, in which $p$ was inserted, may now be feasible to connect with $r$ via the $x'y'-$monotone path $(r,p,q_i), 1\leq i \leq k$, respectively. 
Hence, by Corollary~\ref{cor:xyEdgesChar}, the rooted $x'y'-$MMST of $P$ may change.  

On the other hand, when the $y'$ axis moves away from a line that is perpendicular or parallel to a line passing through $p$ and $q$, with $p,q \neq r$, then if the line segment $\overline{pq}$ was previously an edge of the rooted $x'y'-$MMST of $P$, hence it was previously traversed by the $x'y'-$monotone path from $p$ or $q$, say $p$, to $r$, now might not be feasible to be traversed in the $x'y'-$monotone path from $p$ to $r$ (see Observation~\ref{obs:xyUnnesEdges}).
Hence, the rooted $x'y'-$MMST of $P$ may change.
Furthermore, if one of $p$ or $q$ coincides with $r$, say $p$, then $q$ now does not belong to one of the quadrants that it belonged previously. 
Hence, some points of $P\setminus\{r,q\}$ which previously were adjacent to $q$ and previously belonged to the same quadrant with $q$ are now on different quadrants.
Hence, from Corollary~\ref{cor:pointsOnDifQuadr} it follows that the edges that previously connected $q$ and these points cannot belong to the current rooted $x'y'-$MMST.

At no other moment during the rotation of the $x'y'$ Cartesian System, i.e.{} when the $y'$ axis does not reach (move away from) a line that is perpendicular or parallel to a line passing through two points of $P$, may become feasible/infeasible (while previously it was infeasible/feasible) to traverse a line segment connecting two points of $P$ in a $x'y'-$monotone path from $r$ to a point of $P\setminus\{r\}$.
From the previous sentence and Corollary~\ref{cor:xyEdgesChar}, it follows that at no other moment during the rotation of the $x'y'$ Cartesian System may the rooted $x'y'-$MMST of $P$ change.
\end{proof}

\begin{obs}\label{obs:quarterPlaneRot}
Let $x' , y'$ and $x'', y''$ be axes of two different Cartesian Systems s.t.{} $x''$ (resp.~$y''$) forms with $x'$ (resp.~$y'$) a counterclockwise angle equal to $k\frac{\pi}{2}, k = 1, 2, 3$. 
Then, the rooted $x'y'-$MMST of $P$ coincides with the rooted $x''y''-$MMST of $P$.
\end{obs}

Based on the previous Observations, we define the set $\Theta = \{\theta \in [0, \frac{\pi}{2}):$ a line of slope $\theta$ is either perpendicular or parallel to a line connecting two points of $P\}$. 
 Let $S[\Theta]$ be the sorted sequence that contains the slopes in $\Theta$ in increasing order, i.e.{} $S[\Theta] = (\theta_0, \theta_1, \ldots, \theta_{m-1})$, $\theta_i < \theta_{i+1}, 0 \leq i \leq m-2$  and $ m \leq \binom{|P|}{2}$.
 Then, we define the critical set of slopes $\Theta_{\text{critical}} = \{\theta_0$, $\theta_1$, \ldots, $\theta_{m-1}\} \cup \{\frac{\theta_0 + \theta_1}{2}$, $\frac{\theta_1+ \theta_2}{2}$, \ldots, $\frac{\theta_{m-2}+ \theta_{m-1}}{2}$, $\frac{\theta_{m-1}+ \frac{\pi}{2}}{2}\}$.
We now ``name'' the Cartesian Systems such that their vertical axis has slope in $\Theta_{\text{critical}}$.
More formally, let $x_0y_0$, $x_1y_1$, \ldots, $x_{2m-1}y_{2m-1}$ be the Cartesian Systems such that $y_{2i}$ has slope $\theta_i, 0 \leq i \leq m-1$, $y_{2i+1}$ has slope $\frac{\theta_i + \theta_{i+1}}{2}, 0 \leq i \leq m-2$ and $y_{2m-1}$ has slope $\frac{\theta_{m-1} + \frac{\pi}{2}}{2}$. 

In analogy with Lemma~\ref{lem:MonotoneNecessaryDirections} we obtain the following Lemma.
\begin{lemma}\label{lem:2Dmonotone_Directions}
The rooted 2D-UMMST of $P$ is one of the rooted $x_iy_i-$MMST of $P$, $i =0, 1, 2, \ldots, 2m-1$, and more specifically the one of minimum cost.
\end{lemma}

\begin{theorem}
We can produce the rooted 2D-UMMST of a rooted point set $P$ in $O(|P|^2 \log |P|)$ time.
\end{theorem}

\begin{proof}
We give our algorithm which produces the rooted 2D-UMMST of $P$.
Our rooted 2D-UMMST algorithm is a rotational sweep algorithm analogous to the rooted UMMST algorithm given in Subsection~\ref{subsec:UMMST}.
It rotates a Cartesian System $x'y'$ which initially coincides with $x_0y_0$ counterclockwise until it coincides with the Cartesian System $xy$, i.e.{} the given Cartesian System.
Throughout this rotation, it updates the rooted $x'y'-$MMST of $P$.
By Lemma~\ref{lem:2Dmonotone_Directions}, we only need to compute the rooted $x_iy_i-$MMST of $P$ for the Cartesian Systems $x_iy_i$, $i = 0, 1$, \ldots, $2m-1$. 
Let $T_{i}^{\text{opt}}$ be the rooted $x_iy_i-$MMST of $P$, $i =$ $0$, $1$, \ldots, $2m-1$.
Then, our rooted 2D-UMMST algorithm is restated as follows.
The algorithm initially computes the Cartesian Systems $x_0y_0$, $x_1y_1$, \ldots, $x_{2m-1}y_{2m-1}$ in $O(|P|^2\log |P|)$ time.
Then, it constructs $T_{0}^{\text{opt}}$ using Theorem~\ref{thm:minSpanningXYrooted}.
Then, it iterates for $i =1$, $2$, \ldots, $2m-1$, obtaining $T_{1}^{\text{opt}}$, $T_{2}^{\text{opt}}$, \ldots, $T_{2m-1}^{\text{opt}}$ in this order.
Throughout its execution it stores the Cartesian System min$X$min$Y$ in which it encountered the minimum cost solution found so far.
Finally, it returns the rooted min$X$min$Y-$MMST of $P$, which is recomputed using Theorem~\ref{thm:minSpanningXYrooted}.

From Lemma~\ref{lem:2Dmonotone_Directions}, our rooted 2D-UMMST algorithm produces the rooted 2D-UMMST of $P$.
We only need to analyze its time complexity.
We use similar data structures PD$(p), p \in P\setminus\{r\}$, to the ones employed in Theorem~\ref{thm:spanningTreeOneDir}.
More specifically, for each point $p \in P \setminus \{r\}$ the data structure PD$(p)$ is a self-balancing binary search tree that contains the pairs $(q,d(p,q))$ for all the points $q \in P\setminus\{p\}$ such that $\overline{pq}$ can be traversed in a $x_iy_i-$monotone path from $p$ to $r$ (i.e.{} for the points $q \in P\setminus\{p\}$ that lie on the same quadrant with $p$ w.r.t.{} the Cartesian System $x_iy_i$ and $|x_i(q)| \leq |x_i(p)|$ and $|y_i(q)| \leq |y_i(p)|$), where $i$ is the index of the current iteration of our algorithm.
Then, using similar arguments to the arguments employed in Theorem~\ref{thm:spanningTreeOneDir}, the time complexity of the algorithm is $O(|P|^2 \log |P|)$. 
\end{proof}

\subsection{Recognizing Rooted Uniform 2D-monotone Graphs}

We now study the problem of recognizing if a given rooted connected geometric graph $G = (P, E)$, with root $r$, is rooted uniform 2D-monotone.
Our approach is analogous to the approach we took for recognizing rooted uniform monotone graphs in Subsection~\ref{subsec:recUM}.

For each $p \in P\setminus \{r\}$ let $A(p,x,y)$ be the set $\{q: $ $q\in $Adj$(p)$ and $q$ lies on the same quadrant of the plane with $p$ and $|x(q)| \leq |x(p)|$ and $|y(q)| \leq |y(p)|$ $\}$.
Let $B(x,y)$ be the set $\{p: p\in P\setminus \{r\}$ and $A(p,x,y) \neq \o$ $\}$.
Then, similarly to Lemma~\ref{lem:charYrootSetTerm} we obtain the following Lemma.

\begin{lemma}
\label{lem:2DmonotneChar}
$G$ is rooted $xy-$monotone if and only if $|B(x,y)|$ equals to $|P|-1$.
\end{lemma}

\begin{remark} 
The previous Lemma implies a $O(|E|)$ time recognition algorithm, that decides if $G$ is rooted $xy-$monotone, different from the algorithm given in Theorem~\ref{thm:recRootedXY}.
\end{remark}

Similarly to Observation~\ref{obs:rotationEventsXYproduction}, we obtain the following Observation.

\begin{obs}
If we rotate a Cartesian System $x'y'$ counterclockwise then the sets $A(p,x',y'), p \in P\setminus \{r\}$ and $B(x',y')$ change only when the $y'$ axis reaches (or moves away from) a line perpendicular or parallel to an edge of $G$ or when the $y'$ axis reaches (or moves away from) a line perpendicular or parallel to a line connecting $r$ with another point of $P$.
\end{obs}

From the previous Observation, it follows that we need to take into account only the Cartesian Systems $x_0y_0$, $x_1y_1$, \ldots, $x_{2m-1}y_{2m-1}$, $m < |E| + |P|$ such that $y_{0}, y_2, \ldots, y_{2m-2}$, are all the axes that are either (i) perpendicular or parallel to some edge of $E$ or (ii) perpendicular or parallel to some line connecting $r$ with another point in $P$. 
The slope of each $y_{2i}$ is $\theta_i, 0 \leq i \leq m-1$ and it holds that $0 \leq \theta_0 < \theta_1 <$ \ldots $< \theta_{m-1} < \frac{\pi}{2}$. 
Moreover, the slope of each $y_{2i+1}$ is equal to $\frac{\theta_{i} + \theta_{i+1}}{2}, i = 0, 1, \ldots, m-2$ and the slope of $y_{2m-1}$ is equal to $\frac{\theta_{m-1} + \frac{\pi}{2}}{2}$.
Similarly to Lemma~\ref{lem:recArbNecDirections}, we obtain the following Lemma.

\begin{lemma}\label{lem:sufDirections2DUni}
$G$ is rooted uniform 2D-monotone if and only if it is rooted $x_iy_i-$monotone for some Cartesian System $x_iy_i$, $i = 0, 1, \ldots, 2m-1$.
\end{lemma}

\begin{theorem}
Given a rooted connected geometric graph $G = (P, E)$, we can decide in $O(|E|\log |P|)$ time if $G$ is rooted uniform 2D-monotone. 
\end{theorem}

\begin{proof}
The proof is similar to the proof of Theorem~\ref{thm:recRootedMonotone}.
We employ a rotational sweep algorithm that decides if $G$ is rooted uniform 2D-monotone. 
From Lemma~\ref{lem:sufDirections2DUni}, our rooted uniform 2D-monotone recognition algorithm decides if $G$ is rooted uniform 2D-monotone by testing if $G$ is rooted $x_iy_i-$monotone for some $i = 0,1, \ldots, 2m-1$.
It tests that in this order, i.e.{} it first checks $x_0y_0$ then $x_1y_1$, \ldots, and at the end it checks $x_{2m-1}y_{2m-1}$.

We now show that its complexity is $O(|E| \log |P|)$.
We can compute the Cartesian Systems $x_0y_0$, $x_1y_1$, \ldots, $x_{2m-1}y_{2m-1}$ in $O(|E| \log |P|)$ time.  
The algorithm maintains for each $p \in P\setminus\{r\}$ a data structure $A(p)$ which represents the $A(p,x_i,y_i)$ (when the algorithm checks the $x_iy_i$ Cartesian System) and can be implemented as a $2-3$ tree that stores the indices of the points that it contains. 
Moreover, the algorithm maintains a data structure $B$ that represents the $B(x_i,y_i)$ (when the algorithm checks the $x_iy_i$ Cartesian System) and is implemented as an array of boolean of $O(|P|)$ size.
Using similar analysis to the one presented in Theorem~\ref{thm:recRootedMonotone}, the initial construction of all $A(p), p \in P\setminus\{r\}$ s.t.{} $A(p)$ equals to $A(p,x_0,y_0)$ takes $O(|E|\log |P|)$ total time.
Then, the construction of $B$ s.t.{} $B$ equals to $B(x_0,y_0)$ takes $O(|P|)$ time.
Furthermore, using similar arguments to the ones presented in Theorem~\ref{thm:recRootedMonotone}, the updates of all $A(p), p \in P\setminus\{r\}$, and $B$ throughout all the execution of the algorithm take $O(|E|\log |P|)$ total time.
Additionally, from Lemma~\ref{lem:2DmonotneChar}, given $B$ equal to $B(x_i,y_i)$, it can be decided in $O(1)$ time if $G$ is rooted $x_iy_i-$monotone.
Hence, performing all the tests, i.e.{} if $G$ is rooted $x_iy_i-$monotone for $i = 0$, $1$, \ldots, $2m-1$, using the data structure $B$ take $O(|E|)$ total time.
From all the previous, it follows that the time complexity of the algorithm is $O(|E|\log |P|)$. 
\end{proof}

 \section{Conclusions and Future Work}
In this article we studied the problem of constructing the minimum cost spanning geometric graph of a given rooted point set in which the root is connected to all other vertices by paths that are monotone w.r.t.{} a single direction, i.e.{} they are $y$-monotone (or w.r.t.{} a pair of orthogonal directions, i.e.{} they are $xy$-monotone). 
We showed that the minimum cost spanning geometric graph is actually a tree and we proposed polynomial time algorithms that construct it for the case where the direction (the pair of orthogonal directions) of monotonicity is given or remains to be determined.

Several directions for further research are open.

\begin{enumerate}
\item
We studied rooted point sets and we built the minimum cost spanning tree that contains monotone paths w.r.t.{} a single direction from the root $r$ to any other point in the point set.
What about the case where we are given a $k$-rooted point set, i.e.{} a set with $k$ designated points as its roots, and we are asked to find the minimum cost spanning geometric graph containing monotone paths w.r.t.{} a single direction from each root to every other point in the point set.
In this case, is a wanted graph a tree and additionally, can we find a polynomial time algorithm for this problem?

In the extreme case where all points in the point set $P$ are designated as roots, the problem is trivial. 
Since a geometric graph $G=(P,E)$ is $y-$monotone only if it contains as subgraph the graph path $W_{|P|}$ visiting all points in increasing order of their $y$ coordinates~\cite{Ang15}, the $y-$monotone minimum spanning tree of $P$ is actually the graph path $W_{|P|}$.
Furthermore, the uniform monotone minimum spanning tree of $P$ can be efficiently produced by a rotational sweep algorithm similar to the one employed for the rooted UMMST.

\item We showed that computing the rooted UMMST (or 2D-UMMST) of a rooted point set can be done in polynomial time. 
But is this also the case for the (rooted) monotone (or 2D-monotone) minimum spanning graph of a (rooted) point set or is the problem NP-hard?

\item
We studied the problem of building rooted minimum cost spanning geometric graphs that possess a specific property, and we focused on the property of monotonicity (w.r.t.{} one or two orthogonal directions).
What if we consider a different requirement/property?
For example, we can ask for the minimum cost spanning geometric graph containing increasing-chord paths or self-approaching paths (see~\cite{IckKL99,AlaCGLP13}) from the root to any other point in the point set.
In this case, is the sought graph a tree and does there exist an efficient algorithm that produces it?
\end{enumerate}

\bibliographystyle{plain}
\bibliography{ic-sa,templib,monotone,upward,pointLoc,NN}
\end{document}